\documentclass[11pt,a4paper]{article}
\usepackage[utf8]{inputenc}
\usepackage[english]{babel}
\usepackage{amsmath,amsfonts,amssymb,amsthm,graphicx,lmodern}
\usepackage[left=2.5cm,right=2.5cm,top=2.5cm,bottom=2.5cm]{geometry}

\newtheorem{asmpt}{Assumption}
\newtheorem{thm}{Theorem}

\author{Máté Csanád\\Eötvös University, Budapest, Hungary}
\title{Is there a physical continuum?}
\date{\today}

\begin{document}

\maketitle

\abstract{We are used to the fact that most if not all physical theories are based on the set of real numbers (or another associative division algebra). These all have a cardinality larger than that of the natural numbers, i.e. form a continuum. It is often asked, whether there really is a continuum in the physical world, or whether a future physical theory could work with just countable infinities. The latter could for example be compatible with a quantized space-time. In this paper we formulate a simple model of the brain and show that within the presented natural assumptions, the continuum has to exist for at least some physical quantities.}

\section{Introduction}

Real numbers were first discussed by Descartes, who noted that some polinomials have roots that are not rational.~\cite{Domski:2021} Later it was also proven that real numbers (members of the set $\mathbb{R}$) are either rational (members of the set $\mathbb{Q}$) or irrational, and the cardinality of the set of rational numbers is countable ($\aleph_0$), while that of the set of irrational numbers is uncountable and equal to $\mathfrak{c}$, the cardinality of $\mathbb{R}$. The continuum hypothesis states that there is no subset of $\mathbb{R}$ that is of cardinalty strictly smaller than $\mathfrak{c}$ and strictly larger than $\aleph_0$. This can not be proven within the framework of the Zermelo--Fraenkel set theory with the axiom of choice (ZFC), or in other words, it is independent from it.~\cite{Godel:1940,Cohen:1963,Cohen:1964} One has to note however the following philosophical problem with cardinalities. According to Skolem's paradox, every countable axiomatization of set theory has a countable model, even though there exist not countable sets in the given axiomatization. The resolution~\cite{Putnam:1980} of this paradox is that cardinality is dependent on the axiomatization, as bijections between two sets may exist in one axiomatization, while there are none in another one.  

While the above points are interesting from a pure mathematical point if view, the discovery of real numbers (the continuum) was of probably even greater importance in physics, since it made differential calculus possible. Since Leibniz and Newton, possibly all physical theories are based on the real numbers, or an extension of them, the complex numbers or the quaternions (and any associative division algebra is isomporph to one of these according to the Frobenius theorem), maybe the octonions or some non-associative division algebra.\cite{Baez:2002}

On the other hand, the advent of quantum physics showed that many complications can be solved if even space-time was quantized, i.e. if theories would not be based on the continuum~\cite{Hartland:1947}. One can then ask the question, if there are any arguments beyond present theories (such as quantum field theories or classical gravity, see e.g. Refs.~\cite{Mann:2020,Siegel:2018}) that would help us decide if the continuum is part of the physical reality or just a mathematical framework. In this paper we present an argument that favors the physical existence of the continuum, without relying on a specific physical theory.

\section{An extremely simple model of the brain}

Let us suppose the brain can be modeled with a set of ingredients (cells, neurons, molecules, atoms, particles, etc.), and each of these can be in one of its possible states. One may realize that this is a pretty simple and general model, and in fact most finite physical objects can be modeled such a way. Now the numbers of ingredients shall be finite (let us call it $N$), for example due to the finite amount of available energy or mass. In principle one could imagine a model with for example an infinite amount of photons or gluons in the brain, with their energies following a sequence that creates a finite sum. In this paper however, following Occam's razor, we make the assumption that the number of brain ingredients is finite.

\begin{asmpt}
The brain can be modeled by a finite set of ingredients. The possible states of any of the ingredients can be known a priori.
\end{asmpt}
Let us declare the ``brain ingredients'' as

\begin{align}
p=\{p_i \; | \; i\in \mathbb{N} \, \& \, i<N \}.
\end{align}
Let us similarly declare the set of possible states of a given ingredient as

\begin{align}
s=\{s_i \; | \; i\in \mathbb{N} \, \& \, i<N \}.
\end{align}
The cardinality of $s_i$ (for any $i$) can be
\begin{itemize}
\item finite,
\item countably infinite ($\aleph_0$),
\item or larger (at least $\mathfrak{c}$),
\end{itemize}
and this could differ from ingredient to ingredient. Let us adopt a reductionist approach here and determine the number of possible states {\emph of the whole brain}. Let us declare $S$ as the set of possible states of the whole brain and make the following assumption.
\begin{asmpt}
The possible states of the brain can be constructed via the possible states of its ingredients, i.e. every distinct brain state can be described as $(\sigma_1,\dots,\sigma_N)$ where $\sigma_i \in s_i$.
\end{asmpt}
This simple model is illustrated in Fig.~\ref{f:brain}.

\begin{figure}
\begin{center}
\includegraphics[width=0.6\linewidth]{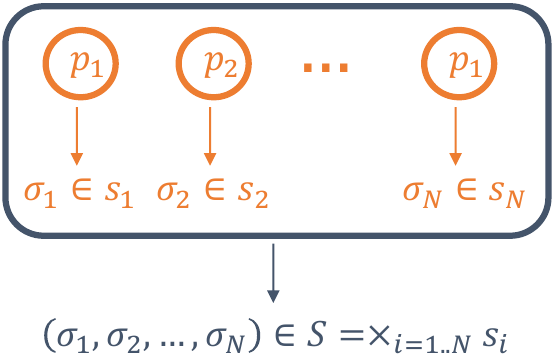}
\end{center}
\caption{The simple model of the brain outlined in this paper. Ingredients are labeled with $p_i$, their possible set of states are $s_i$, and one selected state is $\sigma_i$ (for $i=1\dots N$). The state of the brain is then an element of $\underset{i=1\dots N}{\times} s_i$, in this case $(\sigma_1,\dots,\sigma_N)$, where $\sigma_i \in s_i$, as indicated in the text. \label{f:brain}}
\end{figure}  

With this, we can make the following statements:
\begin{thm}
$\forall i=1\dots N: \; |s_i|\leq\mathfrak{c} \; \& \; \exists i\in \mathbb{N}$ such that $|s_i|=\mathfrak{c} \Rightarrow |S|=\mathfrak{c}$
\end{thm}
\begin{proof}
$S$, based on the axiom of choice, can be constructed as $s_1\times \cdots \times s_N \equiv \underset{i=1\dots N}{\times} s_i$, as one ``brain state'' corresponds to the state of all the ingredients. It is furthermore known that for any infinite set $X$, $|X^N|=|X|$. On the other hand for sets $X$ and $Y$, if $|X|\leq\mathfrak{c}$ and $|Y|=\mathfrak{c}$, then $|X\times Y|=\mathfrak{c}$. Putting all this together it is clear that our statement is proven.
\end{proof}
\begin{thm}
$\forall i=1\dots N: \; |s_i|\leq\aleph_0 \Rightarrow |S|\leq\aleph_0$\label{t:scountable}
\end{thm}
\begin{thm}
$\forall i=1\dots N: \; |s_i|$ is finite $\Rightarrow |S|$ is also finite.
\end{thm}
\begin{proof}
These latter two statements can be proven similarly to the previous proof, taking also into account that for finite sets, the cardinality of the Descartes-product is the product of the cardinalities.
\end{proof}

Let us make a small detour from the above, and investigate what would happen if the number of ingredients would be countably infinite (i.e. if $|p|=\aleph_0$ and of course also $|s|=\aleph_0$). For this case a different statement could be proven:
\begin{thm}
$|p|=\aleph_0 \; \& \; \forall i\in \mathbb{N}: \; |s_i| \textnormal{ is finite } \Rightarrow |S|=\mathfrak{c}$.
\end{thm}
\begin{proof}
Let us again recall that the above statement refers to the case when we depart from our original assumption of $|p|=N$. It us furthermore known that $\{0,1\}^\mathbb{N}=\mathfrak{c}$, or more generally, the countably infinite Descartes product of finite sets (with at least two members) has the cardinality $\mathfrak{c}$. This can be proven via Cantor's method: if one enumerates (``counts'') the members of this set, then one can construct a member that was not on the enumeration list, hence the product set has to have a cardinality of at least $\mathfrak{c}$. Furthermore, a bijection can be created from such sets to a subset of the reals, hence the cardinality of the discussed set is at most $\mathfrak{c}$. From this it is clear that $S$ (which is essentially a Descartes product of the $s_i$ sets) would be of cardinality $\mathfrak{c}$ -- if $s$ would be countably infinite.
\end{proof}
This is very important from the perspective of the remaining part of this paper.

\section{How many possible states does the brain have?}
We have proven above (with some assumptions) that the cardinality of possible brain states ($|S|$) can only be $\mathfrak{c}$ if
\begin{itemize}
\item the number of ingredients in the brain is at least $\aleph_0$,
\item or if there is at least one ingredient for which the set of its possible states is also $\mathfrak{c}$.
\end{itemize}
Since we assumed that the former is not true (due to the finite amount of available energy), let us keep only the latter option. Now if generally there is no continuum, and all physical objects are ``quantized'' in a manner that their set of possible states is at most $\aleph_0$, then the brain could at most have $\aleph_0$ possible states. This is essentially equivalent to Theorem 1 outlined above. But do we know anything about the number of possible states of the human brain?

What we {\emph do} know is that if we take the metric space of rational numbers (constructed from $\mathbb{N}$), then this can be completed by means of equivalence classes Cauchy sequences\footnote{Two Cauchy sequences are considered equivalent if the difference of their terms tends to zero.} to explicitly construct the set of real numbers, $\mathbb{R}$. Let us then make two further assumptions.

The first one should be the following:
\begin{asmpt}
Picking any real number, a mathematically trained human could think of one of the Cauchy sequences tending to the given number.
\end{asmpt}
Since the set of real numbers can explicitely be constructed this way, one can find a Cauchy sequence tending to any chosen number, and then one can think of this sequence. Note that one could assume that thinking of any real number is possible, but first that number has to be explicitely constructed, and for that, we need for example the method via Cauchy sequences. While the above is an assumption, we can argue for its reasonable nature as follows. If there is a real number for which one cannot think of any Cauchy sequence tending to it, then this would help us constructing a way to think about that real number. On the other hand, since clearly finite algorithms cannot account for all members of a continuum set, this assumption may be regarded somewhat unnatural. On the other hand, the brain (or some future quantum-like computer) may not be bound by a finite set of symbols or finite algorithms, so the above assumption can actually be fulfilled. In any case, we need to stress at this point that any conclusion we make below depends on the above assumption. 

Our next assumptions are the following:
\begin{asmpt}
There exists a surjection that maps brain states onto thoughts, in other words: different thoughts correspond to different brain states.\label{a:differentbrainstates}
\end{asmpt}
\begin{asmpt}
Thinking of different Cauchy sequences (having different limits) corresponds to different brain states.\label{a:cauchythoughts}
\end{asmpt}
These assumptions really mean that thinking of Cauchy sequences with different limits has to put the brain in different states. These contain a quite natural assumption: if two thoughts would correspond to the same brain state, that would mean that there is something beyond the physical ingredients of the brain, i.e. something {\emph supernatural} would describe our thoughts. Many would agree to the existence of the supernatural, but even so, the brain states could still be in a one-on-one correspondence with the thoughts, given that the brain is still a natural object.

These assumptions, with the previously outlined theorems, lead to the following:
\begin{thm}
$|S|\leq \mathfrak{c}$, or in other words: the number of possible brain states is at least $\mathfrak{c}$. \label{t:brainstatescont}
\end{thm}
\begin{proof}
Based on Assumptions \ref{a:differentbrainstates}-\ref{a:cauchythoughts}, all Cauchy sequences correspond to a separate brain state. If the number of brain states would be $\aleph_0$ or less, then there should be at least one brain state that simultaneously corresponds to multiple different Cauchy sequence, i.e. multiple different thoughts. This means that the brain is in the same state when thinking of Cauchy sequence $a_n$ and when thinking of Cauchy sequence $b_n$. This would violate Assumptions \ref{a:differentbrainstates}-\ref{a:cauchythoughts}.
\end{proof}

Given Theorem \ref{t:brainstatescont}, we can also state this:
\begin{thm}
$\exists i\in \mathbb{N}$ such that $|s_i|\geq\mathfrak{c}$, or in words: there has to be at least one brain ingredient for which the number of possible states has the cardinality $\mathfrak{c}$.
\end{thm}
\begin{proof}
If there would be no such ingredient, then $|S|$ would be at most $\aleph_0$, as outlined in Theorem~\ref{t:scountable}. This would contradict Theorem \ref{t:brainstatescont}.
\end{proof}

\section{Conclusions}
Above we have proven that given our assumptions, there has to be at least one brain ingredient with a continuum set of possible states. This means that the continuum has to exist, at least for some particular physical quantities, observables or objects, if our assumptions are fulfilled. The importance of our argumentation lies in the fact that we did not rely here on any physical theory, such as a quantum field theory or general relativity. We only assumed a simple model of the brain, and made a few further assumptions about brain states and thoughts.  We note that one critical assumption is about the possibility to think of any real number -- finite algorithms clearly cannot do that.

Above we arrived at the conclusion that the possibility of thinking of any possible Cauchy sequence cannot be true unless at least for one ingredient of the brain, the cardinality of the set of its possible states is (at least) $\mathfrak{c}$. This furthermore means that there has to exist some physical quantities which are continuous, i.e. their possible values have (at least) the cardinality $\mathfrak{c}$. Finally, this means that given these specific assumptions, the continuum exists in the physical world as well.

\section*{Acknowledgment}
This research was funded by grants NKFIH FK-123842 and K-138136, as well as 2019-2.1.11-TÉT-2019-00080. The author furthermore acknowledges the insightful discussions on this topic by András László and Péter Ván, as well as the inspiring research of Tamás Matolcsi.

\end{document}